\documentclass[journal,onecolumn]{IEEEtran}
\usepackage{longtable}
\usepackage{multirow}
\usepackage{tabularx}
\usepackage{lipsum}
\usepackage{multicol}
\usepackage{amssymb}
\usepackage{amsmath}
\usepackage{longtable}
\newtheorem{theorem}{Theorem}[section]
\newtheorem{lemma}[theorem]{Lemma}
\newtheorem{proposition}[theorem]{Proposition}
\newtheorem{corollary}[theorem]{Corollary}

\usepackage{xcolor}

\newenvironment{proof}[1][Proof]{\begin{trivlist}
\item[\hskip \labelsep {\bfseries #1}]}{\end{trivlist}}

\newenvironment{remark}[1][Remark]{\begin{trivlist}
\item[\hskip \labelsep {\bfseries #1}]}{\end{trivlist}}

\makeatletter

\newcommand{\Rmnum}[1]{\expandafter\@slowromancap\romannumeral #1@}

\makeatother

\ifCLASSINFOpdf
\else
\fi
\begin{document}

%
\title{Euclidean  and Hermitian LCD MDS codes }

\author{ Claude Carlet$^1$ \and Sihem Mesnager$^2$ \and Chunming Tang$^3$ \and  Yanfeng Qi$^4$
\thanks{This work was supported by SECODE project and
the National Natural Science Foundation of China
(Grant No. 11401480, 11531002). C. Tang
also acknowledges support from 14E013 and
CXTD2014-4 of China West Normal University.
Y. Qi also acknowledges support from Zhejiang provincial Natural Science Foundation of China (LQ17A010008).
}

\thanks{C. Carlet is with Department of Mathematics, Universities of Paris VIII and XIII, LAGA, UMR 7539, CNRS, Sorbonne Paris Cit\'{e}. e-mail: claude.carlet@univ-paris8.fr}

\thanks{S. Mesnager is with Department of Mathematics, Universities of Paris VIII and XIII and Telecom ParisTech, LAGA, UMR 7539, CNRS, Sorbonne Paris Cit\'{e}. e-mail: smesnager@univ-paris8.fr}
\thanks{C. Tang is with School of Mathematics and Information, China West Normal University, Nanchong, Sichuan,  637002, China. e-mail: tangchunmingmath@163.com
}

\thanks{Y. Qi is with School of Science, Hangzhou Dianzi University, Hangzhou, Zhejiang, 310018, China.
e-mail: qiyanfeng07@163.com
}

}

%


\maketitle

\begin{abstract}
Linear codes with complementary duals (abbreviated LCD) are linear codes whose intersection with their dual is trivial.
When they are binary, they play an important role in armoring implementations against side-channel attacks and fault injection attacks. Non-binary LCD codes in characteristic 2 can be transformed into binary LCD codes by expansion. On the other hand, being optimal codes, maximum distance separable codes
(abbreviated MDS) are of much interest from many viewpoints due to their theoretical and practical properties. However,  little work has been done on LCD MDS codes.
In particular, determining the existence of $q$-ary $[n,k]$ LCD MDS  codes for various lengths $n$ and dimensions $k$ is a basic and interesting problem. In this paper, we  firstly study the
problem of the existence of $q$-ary $[n,k]$ LCD MDS  codes and completely solve it for the
Euclidean case. More specifically, we show that for $q>3$ there exists a $q$-ary $[n,k]$ Euclidean LCD MDS code, where $0\le k \le n\le q+1$, or, $q=2^{m}$, $n=q+2$ and $k= 3 \text{ or } q-1$. Secondly, we investigate several constructions of new Euclidean and Hermitian LCD MDS  codes. Our main techniques in
 constructing  Euclidean and Hermitian LCD MDS codes use some linear codes with small dimension or codimension, self-orthogonal codes and generalized Reed-Solomon codes.

\end{abstract}


%
\IEEEpeerreviewmaketitle

\section{Introduction}

A linear  complementary dual  code (abbreviated LCD) is defined as a linear code $\mathcal C$ whose dual code $\mathcal C ^ \perp$ satisfies $\mathcal C \cap \mathcal C^ \perp=\{0\}$.
LCD codes have been widely applied in data storage, communications systems, consumer electronics, and cryptography. In \cite{mas92}, Massey showed that LCD codes  provide an optimum linear coding solution for the two-user binary adder channel. Recently, the first author and Guilley  \cite{CG14} investigated an interesting application of binary LCD codes against side-channel attacks (SCA) and fault injection attacks (FIA) and presented several constructions of LCD codes. They showed in particular that non-binary LCD codes in characteristic 2 can be transformed into binary LCD codes by expansion.  It is then important to keep in mind that, for SCA, the most interesting case is when $q$ is even.

LCD codes are also interesting objects in the general framework of algebraic coding. For asymptotical optimality and bounds of LCD codes, Massey \cite{mas92} showed that there exist asymptotically good LCD codes. Tzeng and Hartmann \cite{TH70} proved that the minimum distance of a class of LCD codes is greater than that given by the BCH bound. Sendrier \cite{sen04} showed that LCD codes meet the asymptotic Gilbert-Varshamov bound using properties of the hull dimension spectrum of linear codes. Dougherty et al. \cite{DKO15} gave a linear programming bound on the largest size of an LCD code of given length and minimum distance. Recently, Galvez et al. \cite{GKL17} studied the maximum minimum distance of LCD codes of fixed length and dimension.

A lot of works have been devoted to the characterization and constructions of LCD codes.  Yang and Massey  provided in \cite{YM94} a necessary and sufficient condition under which a cyclic code has a complementary dual. In \cite{GOS16}, quasi-cyclic codes that are LCD have been characterized and studied using their concatenated structures. Criteria for complementary duality of  generalized quasi-cyclic codes (GQC) bearing on the component codes are given and some explicit long GQC that are LCD, but not quasi-cyclic, have been exhibited in \cite{OOS17}.  In \cite{DNS17}, Dinh, Nguyend and Sriboonchitta  investigated the algebraic structure of $\lambda$-constacyclic codes over  finite commutative semi-simple rings. Among others, necessary and sufficient conditions for the existence of  LCD,   $\lambda$-constacyclic codes over such finite semi-simple rings have been provided. In \cite{DLL16}, Ding et al. constructed several families of LCD cyclic codes over finite fields and analyzed their parameters. In \cite{LDL16_1} Li et al. studied a class of LCD BCH codes proposed in \cite{LDL16_0} and extended the results on their parameters.
Mesnager et al. \cite{MTQ16} provided a construction of algebraic geometry LCD codes which could be good candidates to be resistant against SCA.  Liu and Liu constructed LCD matrix-product codes using quasi-orthogonal matrices in \cite{LL16}. It was also shown by Kandasamy et al. \cite{KSS12} that maximum rank distance codes generated by the trace-orthogonal-generator matrices are LCD codes. However, little is known on Hermitian LCD codes. More precisely,  it has been  proved in \cite{GOS16} that those codes are asymptotically good \cite{GOS16}. By employing their generator matrices, Boonniyoma and Jitman gave in \cite{BJ16} a sufficient and necessary condition on Hermitian codes for being LCD. Li \cite{Li17} constructed some cyclic Hermitian LCD codes over finite fields and analyzed their parameters.

MDS codes are of significantly practical and theoretical interest, since, for a fixed length and dimension, they have the largest error correcting and detecting capabilities and their weight distribution is known.
Thus, it is natural to consider the intersection of the classes of LCD codes and MDS codes. One of the central problem in this topic is to determine the existence of $q$-ary LCD MDS  codes for various lengths and dimensions. This problem was completely solved for the Euclidean case when $q$ is even \cite{Jin16}. Below, we summarize the known results on sufficient conditions for the existence of  $q$-ary Euclidean LCD MDS codes of  length $n$ and dimension $k$ when $q$ is odd.

 (i) $n = q + 1$ and $k$ even  with $4 \le k \le n-4$ \cite{Jin16};

(ii) $q$ is an odd square, $n\le \sqrt{q} + 1$ and $0\le k\le n$ \cite{Jin16};

(iii) $q\equiv 1 \mod 4$,
$4\cdot 16^n \cdot n^2<q$ and $0 \le k \le n$ \cite{Jin16};

(iv) $n \mid \frac{q-1}{2}$ or $n\mid \frac{q+1}{2}$ and $0\le k\le n$ \cite{Sar16};

(v)  $n$ even with $n | (q -1)$ and even $k$ with $0\le k \le n$ \cite{ZPS16}.

For Hermitian LCD MDS codes, there are only a few references studying them. In \cite{Sar16},
 a class of $q^2$-ary Hermitian LCD codes with length $q-1$ and dimension $\frac{q-1}{2}$ was given.

The main goal of this manuscript  is to study Euclidean and Hermitian LCD MDS codes. We completely determine all Euclidean LCD MDS for all known MDS parameters.
More precisely, we show that for $q>3$ we have a $q$-ary $[n,k]$ Euclidean MDS code, where $0\le k \le n\le q+1$, or, $q=2^{m}$, $n=q+2$ and $k= 3 \text{ or } q-1$.  Our main techniques are
 based on the construction of  Euclidean and Hermitian LCD MDS codes using some linear codes with small dimension or codimension,
self-orthogonal codes and generalized Reed-Solomon codes.  Some classes of new Euclidean and Hermitian LCD MDS  codes are presented from our constructions.

The paper is organized as follows.
Section \ref{Pre} gives preliminaries and background on MDS codes, Euclidean and Hermitian LCD codes.
In Section \ref{Euc}, we firstly provide a construction of Euclidean LCD code from linear codes. Next, we present some class of Euclidean LCD MDS codes coming from the class of
generalized Reed-Solomon codes. On the basis of these results,  we completely determine all Euclidean LCD MDS codes for all known MDS parameters. In  Section \ref{Her}, we present some class of  Hermitian LCD MDS codes.
\section{Preliminaries}\label{Pre}

Throughout this paper,  $p$ is a prime and $\mathbb F_q$  is the finite field of order $q$, where $q=p^m$ for some positive integer $m$. The set of non-zero elements of $\mathbb F_q$ is denoted by $\mathbb F^\times_q$. For any $x \in \mathbb F_{q^2}$, the conjugate of $x$ is defined as $\overline{x}= x^{q}$. For a matrix $A$,
$A^{T}$  denotes the transposed matrix of matrix $A$ 	and $\overline A$  denotes the conjugate matrix of $A$.  When $A$ is a square matrix, $\text{Spec}(A)$ denotes the set of all eigenvalues of
$A$.
An $[n, k, d]$ linear code $\mathcal C$ over $\mathbb F_q$ is a linear subspace of
$\mathbb F_q$ with dimension $k$ and minimum (Hamming) distance $d$. The value $n-k$ is called the \emph{codimension} of $\mathcal C$.
Given a linear code $\mathcal C$ of length $n$ over $\mathbb F_{q}$, its Euclidean dual code (resp. Hermitian dual code) is denoted by $\mathcal C^ {\perp}$ (resp. $\mathcal C^ {\perp_H}$). The codes $\mathcal C^ {\perp}$  and $\mathcal C^ {\perp_H}$ are defined by
\[\mathcal C^ {\perp}=\{ (b_0, b_1, \cdots, b_{n-1})\in \mathbb F_{q}^n:  \sum _{i=0}^{n-1} b_i  c_i=0, \forall (c_0, c_1, \cdots, c_{n-1}) \in  \mathcal C \}, \]

\[\mathcal C^ {\perp_H}=\{ (b_0, b_1, \cdots, b_{n-1})\in \mathbb F_{q^2}^n:  \sum _{i=0}^{n-1} b_i \overline c_i=0, \forall (c_0, c_1, \cdots, c_{n-1}) \in \mathcal C \}, \]
respectively.

 The minimum distance of an $[n, k, d]$ linear code is bounded by the Singleton bound
\[d\le n+1-k.\]
A  code meeting the above bound is called \emph{Maximum Distance Separable} (MDS).
The following proposition \cite{MS77} will be used in this paper.
\begin{proposition}\label{MDS}
Let $\mathcal C$ be an $[n,k,d]$ codes over $\mathbb F_q$. The following statements are equivalent:

(i) $\mathcal C$ is MDS;

(ii) $\mathcal C^{\perp}$ is MDS;

(iii) $\mathcal C^{\perp_H}$ is MDS;

(iv)  every $k$ columns of a generator matrix $G$ are linearly independent (i.e. any $k$ symbols of the codewords constitute a so-called information set, and can then be taken as message symbols).
\end{proposition}

A linear code $\mathcal C$ over $\mathbb F_{q}$ is called an \emph{LCD code} (or for short, LCD code) if $\mathcal C \cap \mathcal C^ {\perp}=\{0\}$.
A linear code $\mathcal C$ over $\mathbb F_{q^2}$ is called a \emph{Hermitian LCD code} (linear code with Hermitian complementary dual) if  $\mathcal C \cap \mathcal C^ {\perp_H}=\{0\}$. To distinguish between cassical LCD codes and Hermitian ones, we shall precise \emph{Euclidean LCD code} in the former case.
The following  proposition gives a complete characterization of
 Euclidean and Hermitian LCD codes (see. \cite{BJ16,CG14}).
\begin{proposition}\label{LDC}
 If $G$ is a generator matrix for the $[n, k]$ linear code $\mathcal C$, then $\mathcal C$ is an Euclidean (resp. a Hermitian) LCD code if and only if the $k \times k$ matrix $G G^{T}$
 (resp. $G\overline{G}^{T}$) is nonsingular.
\end{proposition}

An $[n, k]$ linear code $C$ is said to be \emph{Euclidean self-orthogonal} (resp. \emph{Hermitian self-orthogonal})  if
$  \mathcal C \subseteq
\mathcal C^{\perp}$ (resp. $\mathcal C \subseteq
\mathcal C^{\perp_H} $). And $C$ is said to be \emph{Euclidean self-dual} (resp. \emph{Hermitian self-dual}) if
$  \mathcal C =
\mathcal C^{\perp}$ (resp. $\mathcal C =
\mathcal C^{\perp_H} $). It is easy to see that any self-orthogonal (resp. self-dual)
code has dimension $k\le \frac{n}{2}$ (resp. $k= \frac{n}{2}$).

The following  well-known result gives a complete characterization of
 Euclidean and Hermitian self-orthogonal codes.
\begin{proposition}\label{SOC}
 If $G$ is a generator matrix for the $[n, k]$ linear code $\mathcal C$, then $\mathcal C$ is an Euclidean (resp. a Hermitian) self-orthogonal code if and only if the $k \times k$ matrix $G G^{T}$ (resp. $G\overline{G}^{T}$) is the zero matrix.
\end{proposition}

Any $[n,k]$ linear code over a field is equivalent to a code generated by a matrix of the form $[I_k: P]$ where $I_k$ denotes the $k\times k$ identity matrix.

\section{Existence and constructions of  Euclidean LCD  MDS codes}\label{Euc}
In this section, we study the existence of Euclidean LCD MDS codes and exhibit some constructions of those codes. We  completely determine all [n,k] Euclidean LCD MDS codes over $\mathbb F_q$ with $0\le k \le q+1$.
\begin{proposition}\label{LCD-G}
Let $G=\begin{bmatrix}
 I_k:  P
\end{bmatrix}$ be the generator matrix of an
$[n,k,d]$ linear code $\mathcal{C}$.  Let $\alpha\in
\mathbb{F}_{q}^{\times}$ and
$\mathcal{C}_{\alpha}$ be the linear code
generated by the matrix
$\begin{bmatrix}
 I_k: \alpha P
\end{bmatrix}$.
Then $\mathcal{C}_{\alpha}$ is
an $[n,k,d]$ LCD code if and only if
$-\frac{1}{\alpha^2} \not \in \text{Spec}(PP^T)$.
\end{proposition}
\begin{proof}
By setting $G_{\alpha}=[I_k: \alpha P]$, one has
\begin{align*}
G_{\alpha}G_{\alpha}^T=&
\begin{bmatrix}
 I_k: \alpha P
\end{bmatrix}
\begin{bmatrix}
I_k\\
\alpha P^T
\end{bmatrix}\\
=& 
I_k+\alpha^2PP^T.
\end{align*}
The matrix $G_{\alpha}G_{\alpha}^T$
is nonsingular if and only if
$-\frac{1}{\alpha^2} \not \in \text{Spec}(PP^T)$.
The result follows from Proposition \ref{LDC}.
\end{proof}

\begin{proposition}\label{E-LCD}
Let $k$ be an  integer such that
$0\leq k \leq q-2$ for $p=2$ or
$0\leq k \leq \frac{q-3}{2}$ for
odd prime $p$. Let $\mathcal{C}$ be
an $[n,k,d]$ linear code with  generator matrix
$G=\begin{bmatrix} I_k: P\end{bmatrix}$. Then there exists
 $\alpha\in \mathbb{F}_q^{\times}$
such that the linear
code $\mathcal{C}_{\alpha}$ generated by the matrix $G_{\alpha}
=\begin{bmatrix} I_k: \alpha P\end{bmatrix}$
is an LCD code.
\end{proposition}
\begin{proof}Recall that $\text{Spec}(PP^T)$ has at most $k$ elements.
Let $S=\{\alpha\in \mathbb{F}_q^{\times}
: -\frac{1}{\alpha^2} \in \text{Spec}(PP^T)\}$.

When $p=2$,  we have $\#S\leq k\leq q-2<\#
\mathbb{F}_q^{\times}$. Hence,
$\mathbb{F}_q^{\times} \backslash S
\neq \emptyset$.
Take any $\alpha \in \mathbb{F}_q^{\times} \backslash S$.
From Proposition \ref{LCD-G}, the linear code
$\mathcal{C}_{\alpha}$ is an $[n,k,d]$ LCD code.

When $p$ is odd, we have  $\#S\leq 2k\leq q-3<\#
\mathbb{F}_q^{\times}$.
Hence $\mathbb{F}_q^{\times} \backslash S
 \neq \emptyset$.
 Take any $\alpha \in \mathbb{F}_q^{\times} \backslash S$. From Proposition \ref{LCD-G}, the linear code
$\mathcal{C}_{\alpha}$ is an $[n,k,d]$ LCD code.
\end{proof}

\begin{theorem}\label{k,n-k}
Assume that there exists an $[n,k,d]$ linear code $\mathcal C$ with
 $0\le k\le q-2$ or $0\le n-k\le q-2$ for $p=2$, and $0\le k\le \frac{q-3}{2}$ or $0\le n-k\le \frac{q-3}{2}$ for odd $p$. Then there exists an LCD code having the same parameters as $\mathcal C$.
\end{theorem}

\begin{proof}
If $0\le k\le q-2$ for $p=2$ or $0\le k\le \frac{q-3}{2}$ for odd $p$, the result follows from Proposition \ref{E-LCD}.\\
Assume $0\le n-k\le q-2$ for $p=2$ or $0\le n-k\le \frac{q-3}{2}$ for odd $p$. Let, without loss of generality (up to permutation of the coodinates), $\mathcal C$ be an [n,k,d] code with generator matrix $G=[-Q:I_k]$, where $Q$ is a $k\times (n-k)$ matrix. Then, $\mathcal C^{\perp}$
is an [n,n-k] code of generator matrix $[I_{n-k}:Q]$. From Proposition \ref{E-LCD}, there is an $\alpha\in \mathbb F_q^\times$ such that
$[I_{n-k}:\alpha Q]$ generates an LCD code $\mathcal C'$. Hence,  $\mathcal C'^{\perp}$ is an [n,k] LCD code.  Since the  generator matrix of  $\mathcal C'^{\perp}$ is
$[-Q: \frac{1}{\alpha} I_k]$, the minimum distance of $\mathcal C'^{\perp}$ is same as the one of $\mathcal C$. This completes the proof.

\end{proof}

\begin{lemma}\label{SOC-L}
Let $\mathcal{C}$ be an
$[n,k,d]$ linear code generated by the matrix $G=\begin{bmatrix}
I_k: P \end{bmatrix}$. Then
$\mathcal{C}$ is self-orthogonal if and only if
$PP^T=-I_k$.
\end{lemma}
\begin{proof}
Note that
\begin{align*}
GG^T=&
\begin{bmatrix}
I_k: P \end{bmatrix}
\begin{bmatrix}
I_k\\ P^T \end{bmatrix}\\
=& I_k+PP^T.
\end{align*}
From Proposition \ref{SOC},
$\mathcal{C}$ is self-orthogonal if and only if
$PP^T=-I_k$.
\end{proof}

\begin{corollary}\label{SOC2LDC}
Let $\mathcal{C}$ be an
$[n,k,d]$ self-orthogonal linear code generated by the  matrix $G=\begin{bmatrix}
I_k: P \end{bmatrix}$. Then for any
$\alpha\in\mathbb{F}_q \backslash
\{0,1,-1\}$, the linear
code $\mathcal{C}_{\alpha}$ generated by the matrix  $G_{\alpha}
=\begin{bmatrix} I_k: \alpha P\end{bmatrix}$
is an $[n,k,d]$ LCD code.
\end{corollary}
\begin{proof}
Since $\mathcal{C}$ is a
self-orthogonal code, from Lemma \ref{SOC-L}, $\text{Spec}(PP^T)
=\{-1\}$. From Proposition \ref{LCD-G},
this corollary follows.
\end{proof}

\begin{corollary}\label{k,n-k,MDS}
There exist $[n,k,n-k+1]$ LCD MDS codes
and  $[n,n-k,k+1]$ LCD MDS codes where
 $0\le k\le  n\le q+1$ is such that
$0\leq k \leq q-2$ for $p=2$ and
$0\leq k \leq \frac{q-3}{2}$ for
odd prime $p$.
\end{corollary}
\begin{proof}
Note that there is an $[n,k]$ MDS code for the parameters $n$ and $k$. From Theorem \ref{k,n-k}, there exist
  $[n,k,n-k+1]$ LCD MDS codes. Since the dual of a
LCD MDS code is again an LCD MDS code,
there exist    $[n,n-k,k+1]$ LCD MDS codes.
\end{proof}

\begin{lemma}\label{q+2}
Let $m>1$. Then
there exist LCD MDS codes with
parameters $[2^m+2,3,2^m]$
and $[2^m+2, 2^m-1, 4]$.
\end{lemma}
\begin{proof}
Let $q=2^m$ and $\mathbb{F}_q^{\times}
=\{\alpha_1,\cdots, \alpha_{q-1}\}$, where
$\alpha_1=1$.
Let $A$ be the $3\times (q+2)$ matrix given by
$$
A=
\begin{bmatrix}
1&\cdots & 1& \alpha_2 &0 &0\\
\alpha_1&\cdots & \alpha_{q-1}& 0 &1 &0\\
\alpha_1^2&\cdots & \alpha_{q-1}^2& 0 &0 &1
\end{bmatrix}.
$$
Any three columns of $A$
are linearly independent. Hence,
the linear code
with generator matrix $A$ or parity check matrix
$A$ is an MDS code. Further, we have
$$
AA^T=
\begin{bmatrix}
1+\alpha_2^2 & 0& 0\\
0&1&\beta\\
0&0&1
\end{bmatrix},
$$
where $\beta=\left\{
               \begin{array}{ll}
                 0, & \hbox{$m>2$} \\
                 1, & \hbox{$m=2$}
               \end{array}
             \right.
.$
Note that $\alpha_2^2\neq 1$.
the matrix $AA^T$ is nonsingular.
Hence, the linear code
with  generator matrix $A$ or
 parity check matrix
$A$ is an LCD MDS code.
\end{proof}
\begin{remark}
The unique $[2^1+2,1,4]$ MDS
code is generated by $\begin{bmatrix}
1&1&1&1
\end{bmatrix}$ and is self-orthogonal.
Hence, there do not exist
LCD MDS codes over $\mathbb{F}_2$ with parameters
$[2^m+2,3,2^m]$ or $[2^m+2, 2^m-1,
4]$.
\end{remark}

Let
$\gamma$ be a primitive element of
$\mathbb{F}_q$ and $k$ be a positive integer, with $k|(q-1)$ and
$k<q-1$. Then $\gamma^{\frac{q-1}{k}}$
generates a subgroup with order $k$ of
$\mathbb{F}_q^{\times}$.
Let the subgroup $\langle \gamma^{\frac{q-1}{k}} \rangle
=\{\alpha_0,\cdots, \alpha_{k-1}\}$, where
$\alpha_i\in \mathbb{F}_q^{\times}$ and
$\alpha_0=1$. For any integer $t$, we have
$$
\alpha_0^t+\cdots + \alpha_{k-1}^t
=\left\{
   \begin{array}{ll}
     k, & \hbox{$t\equiv 0 \mod k$;} \\
     0, & \hbox{otherwise.}
   \end{array}
 \right.
$$
Thus, for any $\beta \in \mathbb{F}_q^{\times}$,
we have
\begin{equation}\label{beta-t}
(\beta\alpha_0)^t+\cdots + (\beta\alpha_{k-1})^t
=\left\{
   \begin{array}{ll}
     \beta^tk, & \hbox{$t\equiv 0 \mod k$;} \\
     0, & \hbox{otherwise.}
   \end{array}
 \right.
\end{equation}
Let $A_{\beta}$ be the $k\times k$ matrix given by
\begin{equation}\label{A-beta}
A_{\beta}=\begin{bmatrix}
1 & 1& \cdots& 1&1\\
\beta\alpha_0&\beta\alpha_1&\cdots
&\beta\alpha_{k-2}&\beta\alpha_{k-1}\\
(\beta\alpha_0)^2&(\beta\alpha_1)^2&\cdots
&(\beta\alpha_{k-2})^2&(\beta\alpha_{k-1})^2\\
\vdots&\vdots&\vdots&\vdots& \vdots\\
(\beta\alpha_0)^{k-1}&(\beta\alpha_1)^{k-1}&\cdots
&(\beta\alpha_{k-2})^{k-1}&(\beta\alpha_{k-1})^{k-1}
\end{bmatrix}.
\end{equation}
Then
\begin{equation}\label{AAt}
A_{\beta}A_{\beta}^T
=
\begin{bmatrix}
k&0&\cdots & 0&0\\
0&0&\cdots & 0&\beta^kk\\
0&0&\cdots & \beta^kk&0\\
\vdots&\vdots&\vdots&\vdots& \vdots\\
0&\beta^kk&\cdots & 0&0\\
\end{bmatrix}.
\end{equation}

 When $k\neq \frac{q-1}{2}$,
let $\mathcal{C}$ be the linear code
whose generator matrix is
$$G=\begin{bmatrix}
A_1: A_{\gamma}
\end{bmatrix}.$$
Then $GG^T=A_1A_1^T+A_{\gamma}A_{\gamma}^T$.
From Equation (\ref{AAt}), we have
$$
GG^T
=
\begin{bmatrix}
2k&0&\cdots & 0&0\\
0&0&\cdots & 0&(1+\gamma^k)k\\
0&0&\cdots & (1+\gamma^k)k&0\\
\vdots&\vdots&\vdots&\vdots& \vdots\\
0&(1+\gamma^k)k&\cdots & 0&0\\
\end{bmatrix}.
$$
Hence, $det(GG^T)=2(1+\gamma^k)^{k-1}k^k$.
Note that $k|(q-1)$,
$k\equiv -1\mod p$, $k\neq \frac{q-1}{2}$,
and $\gamma$ is a primitive element
of $\mathbb{F}_q$. Then
$\gamma^k+1\neq 0$ and $det(GG^T)\neq 0$.
Since the matrix $GG^T$ is nonsingular,
$\mathcal{C}$ is an LCD code. Note that
any $k\times k$ submatrix of $G$ is also
nonsingular. Therefore,  $\mathcal{C}$ is an LCD MDS code.

When $q\neq 5$ and $k=\frac{q-1}{2}$,
let $\mathcal{C}$ be the linear code
with the generator matrix
$$G=\begin{bmatrix}
A_1: \gamma A_{\gamma}
\end{bmatrix}.$$
Then $GG^T=
A_1A_1^T+\gamma^2A_{\gamma}
A_{\gamma}^T$. Hence,
$det(GG^T)=(1+\gamma^2)
(1+\gamma^{k+2})^{k-1}k^k$.
If $det(GG^T)=0$, then
$\gamma^2=-1$ or $\gamma^{k+2}
=-1$.
If $\gamma^2=-1$, then $\gamma^4=1$ implying $q=5$.\\
If $\gamma^{k+2}=\gamma^{\frac{q+3}{2}}=-1$, then $\gamma^{q+3}=\gamma^{q-1}\gamma^{4}=\gamma^{4}=1$ implying $q=5$. Contradiction. Hence, $GG^T$ is nonsingular and $\mathcal C$ is an LCD code.

When $q=5$ and $k=\frac{q-1}{2}=2$. Let $\mathcal C$ be the linear code whose generator matrix is given by
\[G=\begin{bmatrix}
1 & 0 & 1 & 1 \\  0 & 1 & 1 &-1
\end{bmatrix}.\]
Then, $\mathcal C$ is an LCD MDS code.
We have the following results.
\begin{lemma}\label{2k}
Let $k$ be a positive integer with $k|(q-1)$ and
$k<q-1$.
Then there exists an
$[2k,k]$ LCD MDS code.
\end{lemma}

We generalize the  construction of
$[2k,k]$ LCD MDS codes to that of
$[2k+1,k]$ LCD MDS codes.
Let $G_1$ be the matrix given by
$$
G_1
=\left\{
   \begin{array}{ll}
     ~[A_1: A_{\gamma}: e_{k-1}], & \hbox{$
k\neq \frac{q-1}{2}$;} \\
     ~[A_1: \gamma A_{\gamma}: e_{k-1}], & \hbox{$
k= \frac{q-1}{2}$ and $q\neq 5$.}
   \end{array}
 \right.
$$
where $e_{k-1}=
[\underbrace{0,0,\cdots,0,1}_{k}]^T$.
Hence, $G_1G_1^T$ is nonsingular and
any $k\times k$ submatrix of $G_1$ is nonsingular.
The linear code $\mathcal{C}_1$
 whose generator matrix $G_1$ is an
$[2k+1,k]$ LCD MDS code.

When $q=5$ and $k=\frac{q-1}{2}=2$. Let $\mathcal C_1$ be the linear code generated by
\[G_1=\begin{bmatrix}
1 & 0 & 1 & 1 & 1\\  0 & 1 & 1 &-1 &2
\end{bmatrix}.\]
Then, $\mathcal C_1$ is an LCD MDS code.
One has the following result.
\begin{lemma}\label{2k+1}
Let $k$ be a positive integer with $k|(q-1)$ and
$k<q-1$.
Then there is an
$[2k+1,k]$ LCD MDS code.
\end{lemma}

Using previous results, we continue the construction of $[2k+2,k]$ LCD MDS codes. Let
$G_2(\alpha,\beta)$ be the
$k\times (2k+2)$ matrix given by
$$
G_2(\alpha,\beta)=
[A_1:\alpha A_{\gamma}: \beta e_0: e_{k-1}],
$$
where $\gamma$ is a primitive
element of $\mathbb{F}_q$,
$\alpha,\beta\in \mathbb{F}_q^{\times}$,
and $e_{0}=
[\underbrace{1,0,\cdots,0,0}_{k}]^T$.
Then
\begin{equation}\label{G-2}
G_2(\alpha,\beta)G_2(\alpha,\beta)^T
=\begin{bmatrix}
(1+\alpha^2)k+\beta^2&0&\cdots &0&0\\
0&0&\cdots &0&(1+\alpha^2\gamma^k)k \\
0&0&\cdots &(1+\alpha^2\gamma^k)k&0\\
\vdots&\vdots&\vdots&\vdots&\vdots\\
0&(1+\alpha^2\gamma^k)k&\cdots &0&1
\end{bmatrix}
\end{equation}
Let $\mathcal{C}_2(\alpha,\beta)$
be the linear code generated by
the matrix $G_2(\alpha,\beta)$.
Note that any $k\times k$ submatrix
of $G_2(\alpha,\beta)$ is nonsingular.
Hence, the code $\mathcal{C}_2(\alpha,\beta)$
is MDS. From Equation (\ref{G-2}),
the linear code
$\mathcal C_2(\alpha,\beta)$ is LCD
if and only if
\begin{equation}\label{alpha-k}
\left\{
  \begin{array}{l}
    (1+\alpha^2)k+\beta^2\neq 0,  \\
    1+\alpha^2\gamma^k\neq 0,
  \end{array}
\right.
\end{equation}

Let $q>3$ and $k\neq \frac{q-1}{2}$.
Take $\alpha=1$, $\beta\in
\mathbb{F}_q^{\times}$, and
$\beta^2\neq -2k$. Equation
(\ref{alpha-k}) is
\begin{equation}\label{alpha-k-1}
\left\{
  \begin{array}{l}
    2k+\beta^2\neq 0,  \\
    1+ \gamma^k\neq 0,
  \end{array}
\right.
\end{equation}
Note that Equation (\ref{alpha-k-1})
holds. Hence, the linear code
$G_2(1,\beta)$
is an LCD MDS code.

Let $q>3$ and $k=\frac{q-1}{2}$.
Equation (\ref{alpha-k}) is
\begin{equation}\label{alpha-k-2}
\left\{
  \begin{array}{l}
    2\beta^2-\alpha^2\neq 1,  \\
    \alpha^2\neq 1
  \end{array}
\right.
\end{equation}
Take $\alpha=\beta=\gamma$. Then, Equation (\ref{alpha-k}) holds. Thus,
the linear code
$\mathcal C_2(\gamma,\gamma)$
is an LCD MDS code.

Let $q=3$. The linear code $\mathcal C_2$
generated by the matrix
$\begin{bmatrix}
1&1&1&1
\end{bmatrix}$
is an $[q+1,\frac{q-1}{2}]$ LCD MDS code.

With the previous discussion, we have the following
lemma.
\begin{lemma}\label{2k+2}
Let $k$ be a positive integer with $k|(q-1)$ and
$k<q-1$.
Then there is an
$[2k+2,k]$ LCD MDS code.
\end{lemma}

\begin{theorem}\label{LCDMDS}
Let $q$ be a prime power with $q> 3$ and
$k,n$ be integers with $0\le k\le n$. Then
there exists a $q$-ary Euclidean LCD MDS code with
parameters $[n,k]$ if one of the following conditions
holds.

{(i)}  $ n\leq q+1$;

{(ii)} $q=2^m$ with positive integer $m$ , $n=q+2$, and
$k=3 \text{ or } q-1$.
\end{theorem}

\begin{proof}
By Lemma \ref{2k}, Lemma \ref{2k+1} and Lemma \ref{2k+2}, $[q-1,\frac{q-1}{2}] $ LCD MDS code, $[q,\frac{q-1}{2}]$ LCD MDS code and $[q+1,\frac{q-1}{2}]$ LCD MDS code exist.
Since the dual of an LCD MDS code is again an LCD MDS code, there are $[q,\frac{q+1}{2}]$ LCD MDS code and $[q+1,\frac{q+3}{2}]$ LCD MDS. From \cite{GA08}, there is an $[q+1,\frac{q+1}{2}]$
sel-dual code MDS code. Thus, when $q\neq 3$, there is an $[q+1,\frac{q+1}{2}]$ LCD MDS code by Corollary \ref{SOC2LDC}. It completes the proof from Corollary \ref{k,n-k,MDS}
and Corollary \ref{q+2}.
\end{proof}
\begin{remark}
When $q=2$, all $2$-ary  LCD MDS codes are  $[n,k]$ codes with $0\le k \le n \le 3$ and $[n,k]\neq [2,1]$.\\
When $q=3$, all $3$-ary LCD MDS codes are  $[n,k]$ codes with $0\le k \le n \le 4$ and $[n,k]\neq [4,2]$.
If the following MDS conjecture holds,  from Theorem \ref{LCDMDS}, we have classified all LCD MDS codes.\\
\textbf{MDS conjecture}:
Let $\mathcal{C}$ be an $[n,k]$
MDS code. Then $n\leq q+1$,
except when $q$ is even and $k\in \{3,q-1\}$,
  in which case $n\leq q+2$.\\

\end{remark}

\section{Existence and construction of  Hermitian LCD MDS codes}\label{Her}
In this section, we introduce some constructions of Hermitian LCD MDS codes and present some  classes of  Hermitian LCD MDS codes.
\begin{proposition}\label{H-LCD-G}
Let  $\mathcal{C}$ be an
$[n,k,d]$ linear code whose generator matrix $G$ is given n
by
$G=\begin{bmatrix}
 I_k:  P
\end{bmatrix}$.  Let $\alpha\in
\mathbb{F}_{q^2}^{\times}$ and
$\mathcal{C}_{\alpha}$ be the linear code
generated by the  matrix
$\begin{bmatrix}
 I_k: \alpha P
\end{bmatrix}$.
Then $\mathcal{C}_{\alpha}$ is
an $[n,k,d]$ Hermitian LCD code if and only if
$-\frac{1}{\alpha \overline{\alpha}} \not \in \text{Spec}(P\overline{P}^T)$.
\end{proposition}
\begin{proof}
Let $G_{\alpha}=\begin{bmatrix}
 I_k: \alpha P
\end{bmatrix}$. Then
\begin{align*}
G_{\alpha}\overline{G}_{\alpha}^T=&
\begin{bmatrix}
 I_k: \alpha P
\end{bmatrix}
\begin{bmatrix}
I_k\\
\overline{\alpha}  \overline{P}^T
\end{bmatrix}\\
=& I_k+\alpha \overline{\alpha}P\overline{P}^T\\
=& (-\alpha\overline{\alpha})(-\frac{1}
{\alpha\overline{\alpha}}I_k-P\overline{P}^T).
\end{align*}
Hence, $G_{\alpha}G_{\alpha}^T$
is nonsingular if and only if
$-\frac{1}
{\alpha\overline{\alpha}} \not \in \text{Spec}(PP^T)$.
From Proposition \ref{LDC},
$\mathcal{C}_{\alpha}$ is a
Hermitian LCD code if and only if
$-\frac{1}{\alpha \overline{\alpha}} \not \in \text{Spec}(P\overline{P}^T)$.
\end{proof}

\begin{proposition}\label{E-H-LCD}
Let $\mathcal{C}$ be an $[n,k,d]$
linear code whose generator matrix is given by
$G=\begin{bmatrix}
 I_k:  P
\end{bmatrix}$.
Let $0\leq k\leq q-2$. Then there
exists an
$\alpha\in
\mathbb{F}_{q^2}^{\times}$  such that
the linear code
$\mathcal{C}_{\alpha}$
generated by the  matrix
$\begin{bmatrix}
 I_k: \alpha P
\end{bmatrix}$  is
an  $[n,k,d]$ Hermitian LCD code.
\end{proposition}
\begin{proof}
Let $S=\{\alpha\in \mathbb{F}_{q^2}^{\times}
: -\frac{1}{\alpha \overline{\alpha}}  \in Spec(P\overline{P}^T)\}$.
The map $\alpha\mapsto \alpha\overline{\alpha}$
from $\mathbb{F}_{q^2}^{\times}$
to $\mathbb{F}_{q}^{\times}$ is a
$q+1$ to $1$ map. Then
$\#S\leq k(q+1) < q^2-1
=\#\mathbb{F}_{q^2}^{\times}$. Hence,
$\mathbb{F}_{q^2}^{\times} \backslash S \neq \emptyset$. Take any $\alpha \in
\mathbb{F}_{q^2}^{\times} \backslash S$.
From Proposition \ref{H-LCD-G},
the linear code
$\mathcal{C}_{\alpha}$
generated by the  matrix
$\begin{bmatrix}
 I_k: \alpha P
\end{bmatrix}$  is
an $[n,k,d]$ Hermitian LCD code.
\end{proof}

\begin{lemma}\label{-I-H}
Let $\mathcal{C}$ be a
$[n,k,d]$ linear code generated by the matrix $G=\begin{bmatrix}
I_k: P \end{bmatrix}$. Then
$\mathcal{C}$ is Hermitian self-orthogonal if and only if
$P\overline{P}^T=-I_k$.
\end{lemma}
\begin{proof}
Note that
\begin{align*}
G\overline{P}^T=&
\begin{bmatrix}
I_k: P \end{bmatrix}
\begin{bmatrix}
I_k\\ \overline{P}^T \end{bmatrix}\\
=& I_k+P\overline{P}^T
\end{align*}
From Proposition \ref{SOC},
$\mathcal{C}$ is Hermitian self-orthogonal if and only if
$P\overline{P}^T=-I_k$.
\end{proof}

\begin{corollary}\label{H-LCD-alpha}
Let $\mathcal{C}$ be an
$[n,k,d]$ Hermitian self-orthogonal code generated by the matrix $G=\begin{bmatrix}
I_k: P \end{bmatrix}$.
Let $\gamma$ be a primitive element
of $\mathbb{F}_{q^2}$.
Then for any $\alpha\in
\mathbb{F}_{q^2}^{\times}\backslash
\{\gamma^{(q-1)i}: 0\leq i\leq q \}$,
the linear code $\mathcal{C}_{\alpha}$
generated by the matrix  $\begin{bmatrix}
 I_k: \alpha P
\end{bmatrix}$  is
an $[n,k,d]$ Hermitian LCD code.
\end{corollary}
\begin{proof}
Since $\mathcal{C}$ is  an
$[n,k,d]$ self-orthogonal code,
$\text{Spec}(P\overline{P}^T)
=\{-1\}$ by Lemma \ref{-I-H}. Note that
$\alpha\overline{\alpha}=1$ if and only if
$\alpha\in
\{\gamma^{(q-1)i}: 0\leq i\leq q \}$.
From Proposition \ref{H-LCD-G}, this corollary follows.
\end{proof}

\begin{corollary}\label{d-cod-MDS}
Let  $0\leq k\leq q-2$ and $0\le k\le n\le q+1$.
Then there exist $q^2$-ary $[n,k,n-k+1]$
Hermitian LCD MDS codes  and $q^2$-ary  $[n,n-k,k+1]$
Hermitian LCD MDS codes.
\end{corollary}
\begin{proof}
Note that there is a $[n,k]$ MDS code for the parameters $n$ an $k$.
From Proposition \ref{E-H-LCD},
there exist $[n,k,n-k+1]$ Hermitian LCD
MDS codes. The Hermitian
dual of a Hermitian LCD MDS code is again a Hermitian LCD MDS code.
Hence,  there exist   $[n,n-k,k+1]$
Hermitian LCD MDS codes.
\end{proof}

From Corollary \ref{d-cod-MDS},
we can construct $q^2$-ary Hermitian LCD MDS codes with dimension or codimension less than $q-1$.
We will present the other construction of  Hermitian LCD MDS codes.

Let $q$ be odd, $k$ be a positive integer and
$\gamma$ a primitive element of
$\mathbb{F}_{q^2}$, where $k|(q^2-1)$,
$k\not | (q+1)$, and
$k\leq \frac{q^2-1}{2}$. Then $\gamma^{\frac{q^2-1}{k}}$
generates a subgroup with order $k$ of
$\mathbb{F}_{q^2}^{\times}$.
Let the subgroup  $\langle \gamma^{\frac{q^2-1}{k}} \rangle
=\{\alpha_0,\cdots, \alpha_{k-1}\}$, where
$\alpha_i\in \mathbb{F}_{q^2}^{\times}$ and
$\alpha_0=1$. For any integer $t$, we have
$$
\alpha_0^t+\cdots + \alpha_{k-1}^t
=\left\{
   \begin{array}{ll}
     k, & \hbox{$t\equiv 0 \mod k$;} \\
     0, & \hbox{otherwise.}
   \end{array}
 \right.
$$
For any $\beta \in \mathbb{F}_{q^2}^{\times}$,
we have
$$
(\beta\alpha_0)^t+\cdots + (\beta\alpha_{k-1})^t
=\left\{
   \begin{array}{ll}
     \beta^tk, & \hbox{$t\equiv 0 \mod k$;} \\
     0, & \hbox{otherwise.}
   \end{array}
 \right.
$$
Let $A_{\beta}$ be the following $k\times k$ matrix
\begin{equation*}
A_{\beta}=\begin{bmatrix}
1 & 1& \cdots& 1&1\\
\beta\alpha_0&\beta\alpha_1&\cdots
&\beta\alpha_{k-2}&\beta\alpha_{k-1}\\
(\beta\alpha_0)^2&(\beta\alpha_1)^2&\cdots
&(\beta\alpha_{k-2})^2&(\beta\alpha_{k-1})^2\\
\vdots&\vdots&\vdots&\vdots& \vdots\\
(\beta\alpha_0)^{k-1}&(\beta\alpha_1)^{k-1}&\cdots
&(\beta\alpha_{k-2})^{k-1}&(\beta\alpha_{k-1})^{k-1}
\end{bmatrix}
\end{equation*}
From
$(\beta\alpha_{l})^{i}(\overline{\beta\alpha_{l}})^{j}
=\beta^{i+qj}\alpha_l^{i+qj}$,
we have
$$
(\beta\alpha_0)^i(\overline{\beta\alpha_{0}})^{j}+\cdots + (\beta\alpha_{k-1})^i(\overline{\beta\alpha_{k-1}})^{j}
=\left\{
   \begin{array}{ll}
     \beta^{i+qj}k, & \hbox{$i+qj\equiv 0 \mod k$;} \\
     0, & \hbox{otherwise.}
   \end{array}
 \right.
$$
Let $a_{\beta}(i,j)$ be the
 entry in the $i$-th row and $j$-th
column of the matrix
$A_{\beta}\overline{A}_{\beta}^T$.
Then we have
\begin{equation}\label{a}
a_{\beta}(i,j)
=\left\{
   \begin{array}{ll}
     \beta^{i+qj}k, & \hbox{$i+qj\equiv 0 \mod k$;} \\
     0, & \hbox{otherwise.}
   \end{array}
 \right.
\end{equation}
Let $\mathcal{C}(\alpha)$ be
a linear code generated by the following matrix
$$
G(\alpha)=
\begin{bmatrix}
A_1: \alpha A_{\gamma}
\end{bmatrix}
$$
where $\alpha=\frac{1}{
\gamma^{\frac{q-1}{2}}\gamma}$.
Obviously, $\mathcal{C}(\alpha)$
is an MDS code.
Let $b(i,j)$ be the
 entry in the $i$-th row and $j$-th
column of the matrix
$G(\alpha)
\overline{G(\alpha)}^T$.
From Equation (\ref{a}), we have
$$
b(i,j)
=\left\{
   \begin{array}{ll}
     (1-\frac{\gamma^{i+qj}}{\gamma^{q+1}})k, & \hbox{$i+qj\equiv 0 \mod k$;} \\
     0, & \hbox{otherwise.}
   \end{array}
 \right.
$$
When $i+qj\equiv 0\mod k$, assume that
$b(i,j)=0$.
Then $\gamma^{i+qj}=\gamma^{q+1}$ and
$i+qj\equiv q+1\mod q^2-1$.
From $k|(q^2-1)$ and $k|(i+qj)$,
we have $k|(q+1)$, which makes a contradiction
with $k\not \mid (q+1)$.
Hence, when $i+qj\equiv 0\mod k$,  we have
$b(i,j)\neq 0$. Note that $k|(q^2-1)$.
Every row of
$G(\alpha)
\overline{G(\alpha)}^T$
has just only a nonzero element. That holds for
each column. Hence, the matrix
$G(\alpha)
\overline{G(\alpha)}^T$ is nonsingular. From Proposition \ref{LDC},
$\mathcal{C}(\alpha)$
is an $[2k,k]$ Hermitian LCD MDS code.

Let $\mathcal{C}_1(\alpha)$ be a linear code generated by the following matrix
$$
G_1(\alpha)=
\begin{bmatrix}
A_1: \alpha A_{\gamma}: e_{k-1}
\end{bmatrix}
$$
where $\alpha=\frac{1}{
\gamma^{\frac{q-1}{2}}\gamma}$
and $e_{k-1}=
[\underbrace{0,0,\cdots,0,1}_{k}]^T$.
From the above discussion, we have that
$\mathcal{C}_1(\alpha)$
is an $[2k+1,k]$ Hermitian
LCD MDS  code.

When $q>2$, there exist an
$\beta \in \mathbb{F}_{q^2}^{\times}$
such that
$\beta\overline{\beta} \neq
-(1-\frac{1}{\gamma^{q+1}})k$.
Take such a $\beta$. Let
$\mathcal{C}_2(\alpha)$  be a linear code
generated by the following matrix
$$
G_2(\alpha)=
\begin{bmatrix}
A_1: \alpha A_{\gamma}: \beta e_0: e_{k-1}
\end{bmatrix}
$$
where $\alpha=\frac{1}{
\gamma^{\frac{q-1}{2}}\gamma}$,
$e_{0}=
[\underbrace{1,0,\cdots,0,0}_{k}]^T$,
and $e_{k-1}=
[\underbrace{0,0,\cdots,0,1}_{k}]^T$.
From the above discussion, we have that
$\mathcal{C}_2(\alpha)$
is an $[2k+2,k]$ Hermitian
LCD MDS code.
With the previous discussion, we have the following
Lemma.
\begin{lemma}\label{3=2k+2}
Let $q$ be odd, $k$ be a positive integer with $k|(q^2-1)$,
$k\not | (q+1)$, and
$k<q^2-1$.
Then there are
$[2k,k]$, $[2k+1]$ and $[2k+2,k]$ LCD MDS code over $\mathbb F_{q^2}$.
\end{lemma}

\begin{theorem}\label{LCDMDS}
Let $q$ be a prime power and
$k,n$ be integers with $0\le k\le n$. Then
there exists a $q^2$-ary Hermitian LCD MDS code with
parameters $[n,k]$ if one of the following conditions
holds.

{(i)}  $ n\leq q+1$, $k\le q-2$ or $n-k \le q-2$;

{(ii)} $q$  odd, $[n,k]\in \{[2k,k], [2k+1,k], [2k+2,k]\}$ where $k$  is a positive integer with $k|(q^2-1)$,
$k\not | (q+1)$, and
$k<q^2-1$.

{(iii)} $q=2^m\ge 8$, $n=q+2$, $k=3$ or $k=q-1$.
\end{theorem}
\begin{proof}
The results follows from Corollary \ref{d-cod-MDS}, Lemma \ref{3=2k+2} and Proposition \ref{E-H-LCD}.
\end{proof}

\section{Concluding Remarks}\label{sect:concluding}
LCD codes have applications in information protection. MDS codes are an important class of linear codes that have found wide applications in both theory and practice.
 Though LCD codes and MDS codes have been extensively studied in literature, there is only  few results on LCD MDS codes.  This paper devoted to the construction of Euclidean and Hermitian LCD MDS codes. We detail some secondary constructions of LCD codes, using linear codes with small dimension and codimension, self-orthogonal codes and generalized Reed-Solomon codes. Some classes of new Euclidean and  Hermitian LCD MDS codes are obtained. Finally, we prove that  there is a $q$-ary $[n,k]$ Euclidean LCD MDS code for any $q>3$ and $0\le k\le n\le q+1$. But it is open whether $q^2$-ary $[n,k]$ Hermitian LCD MDS code exists for all $q>3$ and $0\le k\le n\le q^2+1$. It would be nice if this open problem can be settled.

\end{document}